\documentclass[a4paper,10pt]{article}

\usepackage[utf8]{inputenc}
\usepackage{fullpage}

\usepackage{color}
\usepackage{amsmath,mathtools,amsxtra}
\usepackage{ifthen}
\usepackage{style}
\usepackage{math}
\usepackage[utf8]{inputenc}
\usepackage{fixltx2e,graphicx,mathpazo} 
\usepackage{xspace}

\newboolean{focusModification} 

\definecolor{gray}{RGB}{100,100,100}

\newcommand{\dxinv}{\ensuremath{\partial_x^{-1}}}
\newcommand{\Sing}[1][m]{\ensuremath{M_0^{#1}}\xspace}

\newcommand{\Pp}[1][m]{\ensuremath{M_+^{#1}}\xspace}
\newcommand{\Pm}[1][m]{\ensuremath{M_-^{#1}}\xspace}
\newcommand{\Ppm}[1][m]{\ensuremath{M_\pm^{#1}}\xspace}
\newcommand{\Wp}{\ensuremath{W^l}\xspace}
\newcommand{\Msinh}[1][m]{\ensuremath{M^{#1}_{sh}}\xspace}
\newcommand{\Msin}[1][m]{\ensuremath{M^{#1}_{sg}}\xspace}
\newcommand{\Tr}{\ensuremath{{\it T}}\xspace}
\newcommand{\Fsinh}{\ensuremath{\Psi_{sh}}\xspace}
\newcommand{\Fsinp}{\ensuremath{\Psi_{sg}^+}\xspace}
\newcommand{\Fsinm}{\ensuremath{\Psi_{sg}^-}\xspace}
\newcommand{\Csin}{\ensuremath{c_{sg}}\xspace}
\newcommand{\Csinh}{\ensuremath{c_{sh}}\xspace}
\newcommand{\Hsinh}{\ensuremath{H_{sh}}\xspace}
\newcommand{\Hsin}{\ensuremath{H_{sg}}\xspace}

\usepackage{xparse}
\DeclareDocumentCommand{\Leave}{ O{m} O{}}{\ensuremath{L_{#2}^{#1}}\xspace}

\newcommand{\nbhd}{\mathcal{U}}

\usepackage{units}
\usepackage{amscd}
\usepackage{chngcntr}

 \title{On the integrability of the sine-Gordon equation}
 \author{Yannick Widmer}
 \date{\today}
 
\setboolean{focusModification}{false} 

\begin{document}
\maketitle

\ifthenelse{\boolean{focusModification}}{
  \color{gray}
}{}

\section*{Abstract}

Among other results we show that near the equilibrium point, the Hamiltonian of the sine-Gordon (SG) equation on the circle can be viewed as an element in the Poisson algebra of the modified Korteweg-de Vries (mKdV) equation and hence by well established properties of the latter equation admits Birkhoff coordinates.
On the other hand we prove that there exists a large set of smooth initial data, away from the equilibrium point and lying on the ramification locus of a double cover, for which the initial value problem of the SG equation has no classical solution.

\section{Introduction}
Consider the complex sine-Gordon equation on the circle $\T=\R/\Z$
\begin{equation} \label{eq:complexSine}
		U_{xt}=\sin(U), \quad x\in \T, t \in \R,\; U(t,x)\in\C.
\end{equation}
Solutions of \eqref{eq:complexSine} with values in $\R$, $U(x,t)=u(x,t)\in\R$ are solutions of the real sine-Gordon equation 
\begin{equation}\label{eq:sinGordon}
		u_{xt}=\sin(u), \quad x\in\T,t\in \R,\;u(x,t)\in\R. 
\end{equation} 
whereas imaginary valued solutions, $U(x,t)=\ii u(x,t)\in \ii\R$, are solutions to the sinh-Gordon equation,
\begin{equation}\label{eq:sinhGordon}
		u_{xt}=\sinh(u), \quad x\in\T,t\in \R,\;u(x,t)\in\R.
\end{equation}
 
Equation \eqref{eq:sinGordon} has already been studied extensively
since it arose long ago in the theory of pseudospherical
surfaces (cf. discussion at the end of the introduction).  In view of the many applications it is arguably among the most important equations in contemporary physics:  it  arises in  model field theories (see e.g. \cite{Skyrme1961MesonFields}), superconductivity (see e.g. \cite{Lebwohl1967Superconductivity}) and in
mechanical models of nonlinear wave propagation
(see e.g. \cite{Lamb1971PhysicsOptPulse} and references therein).

One of the main results of this paper says, that away form the ramification locus of a double cover, equation \eqref{eq:sinGordon} can be viewed as a Hamiltonian PDE on the phase space of the focusing modified Korteweg-de Vries (\mkdv) equation
\begin{equation}\label{eq:mkdvfocusing}
    v_t + v_{xxx} + 6v^2v_x = 0, \quad x\in\T,t\in \R,\; v(x,t)\in\R.
 \end{equation}
with Hamiltonian in the corresponding Poisson  algebra of this equation. A similar relation is established between the sinh Gordon equation \eqref{eq:sinhGordon} and the defocusing \mkdv equation
\begin{equation}\label{eq:mkdvdefocusing}
    v_t + v_{xxx} - 6v^2v_x = 0, \quad x\in\T,t\in \R,\; v(x,t)\in\R.
 \end{equation}
Let us first state our results for the sinh-Gordon equation, which are somewhat easier to formulate than the ones for the sine-Gordon equation. Denote $H^m\equiv H^m(\T,\R)$, $m\in\Z_{\geq0}$, the standard Sobolev spaces of real valued function defined on $\T$. Note that any $C^1-$solution $\R\to H^1, t\mapsto u(t,\cdot)$   of \eqref{eq:sinhGordon} satisfies
\[
 \int_\T\sinh(u)\dx=\partial_t \int_\T\partial_x u\dx=0.
\]
This identity suggests to introduce the scale of phase spaces for \eqref{eq:sinhGordon} 
\[
\Msinh=\left\{\, u\in H^{m}\colon\int_\T \sinh(u)\dx=0\,\right\}, \quad m\geq 0.
\]
Note that \Msinh is a real analytic codimension 1 submanifold of $H^{m}$. 
An element of  \\ $H_0^{m+1}=\left\{\; u\in H^{m+1}\; :\; \int_\T u\dx = 0 \; \right\}$ can be written in the form $2\dxinv v$ where $v\in H_0^m$ and $\dxinv$ denotes the mean zero antiderivative.
Furthermore for any $C^1$- functionals $F,G:H^m\to \R$ with $L^2-$gradients in $H^1$, denote by $\{F,G\}$ the Gardner bracket
\[
 \{F,G\}(v) =\int_0^1 \partial_v F\partial_x\partial_v G\dx.
\]
Finally we need to introduce the notion of Poisson algebra of an integrable PDE such as the defocusing \mkdv equation. It is well known that equation \eqref{eq:mkdvdefocusing} admits a Lax pair formulation
\[
	\partial_tL_M = [L_M,B_M].
\]
We say that a Hamiltonian $H$ is in the Poisson algebra of equation \eqref{eq:mkdvdefocusing} if the periodic and anti-periodic spectra of $L_M$ are constants of motion under the flow of $H$.
Our first main result is the following.
\begin{thm}\label{thm:Sinh}
 For any $m\geq 0$, the following holds:
\begin{thmenum}
\item \label{prop:cforSinh}
 The map
 \begin{align*}
 \Fsinh:H_0^{m}&\to \Msinh[m+1],v\mapsto u(v)=2\dxinv v +\Csinh(v) &
  \Csinh(v):=-\atanh\pfrac{\int_\T\sinh(2\dxinv v)\dx}{\int_\T \cosh(2\dxinv v)\dx},
 \end{align*}
is a real analytic diffeomorphism and transforms equation \eqref{eq:sinhGordon}  into
\begin{equation}\label{eq:dxSinh} 
	  v_t=\frac12 \sinh(2\partial_x^{-1}v+\Csinh(v))
\end{equation}
on the space $H_0^m$ where the vector field $\frac12 \sinh(2\partial_x^{-1}v+\Csinh(v))$ is in $H_0^{m+1}$.
 \item  Equation \eqref{eq:dxSinh} is a Hamiltonian PDE with respect to the Gardner bracket with Hamiltonian
 \begin{equation}\label{H:sinh}
 \Hsinh:H_0^m\to\R,v\mapsto \Hsinh(v)= -\frac14\int_\T\cosh (2\dxinv v+\Csinh(v))\dx.
\end{equation}
\item The Hamiltonian \Hsinh is in the Poisson algebra of the defocusing \mkdv equation and hence \eqref{eq:dxSinh} is an integrable PDE on $H^m_0$.
\end{thmenum}
\end{thm}
Theorem \ref{thm:Sinh} combined with the result of \cite{kappeler2008mkdv} on global Birkhoff coordinates for the defocusing \mkdv equation leads to the following application.
\begin{cor}\label{cor:BirkhoffSinh} 
For any $m\geq 0$, equation \eqref{eq:dxSinh} admits global Birkhoff coordinates. In particular, the Hamiltonian \Hsinh is a function of the corresponding actions alone. 
\end{cor}

Our second main result concerns  equation \eqref{eq:sinGordon}. 
It is customary to consider solutions $u$ of \eqref{eq:sinGordon} of the form $u(x+1)=u(x)+2\pi k$. The integer $k$ is referred to as  topological charge (see e.g. \cite{NovikovGrinevich2003topological}). We identify an element $u\in H^{m}$ with the one periodic function on $\R$ $x\mapsto u(x\pmod{2\pi})$ and define the affine spaces
\[
	 \Leave[m][k] =2\pi k x + H_0^{m} + \T_{2\pi}= \left\{ \; u = 2\pi kx + \mathring{u}\;:\; \mathring{u}\in H_0^m + \T_{2\pi}\;\right\}
\] 
where $\T_{2\pi}= \R/2\pi\Z$.
Note that for $u\in \Leave[m][k]$, $m\geq 1$ one has that $\sin(u),\cos(u)\in H^m$ are well defined as well as $\partial_x u = \partial_x\mathring{u}+2\pi k\in H^{m-1}$ where $\partial_x\mathring{u}\in H_0^{m-1}$. Finally, since for $m\geq 1$, the sets \Leave[m][k], $k\in\Z$ are pairwise disjoint, for any solution $u$ of \eqref{eq:sinGordon} the topological charge is a constant of motion and hence it is defined by the initial data. For this reason we fix $k\in\Z$ and drop the index $k$ in the sequel. Note that for  $u = \mathring{u}+2\pi kx \in C^1(\R,\Leave)$, with $\mathring{u}\in C^1(\R,H^m)$ one has for all $t\in\R$ $\partial_t u= \partial_t\mathring{u} +\partial_t2\pi kx =\partial_t\mathring{u} \in H^m$ and hence
\[
	\partial_t\partial_x u=\partial_t\partial_x\mathring{u}=\partial_x\partial_t\mathring{u}
	=\partial_x\left(\partial_t\mathring{u}+\partial_t\left(2\pi kx\right)\right)
	=\partial_x\partial_t u \text{ in } H_0^{m-1}.
\]
This implies that for any $C^1-$ solution {$u:\R\to \Leave$} of \eqref{eq:sinGordon} one has
\[
 \int_\T\sin(u)\dx=\int_\T u_{xt}\dx=\partial_t \int_\T u_x = \partial_t 2\pi k = 0,
\]
hence as a scale of phase spaces of \eqref{eq:sinGordon} we chose
\[
   \Msin=\left\{\, u\in \Leave\colon \int_\T \sin(u)\dx=0\,\right\}.
\]
Note that \Msin is a codimension 1 submanifold of \Leave. Indeed, the $L^2$-gradient of $\int_\T \sin(u) \dx$ is $\cos(u)$ which vanishes identically iff $u\equiv \pi/2$ or $u\equiv 3\pi/2$, implying that on \Msin the $L^2$-gradient of $\int_\T\sin(u)\dx$ doesn't vanish.
For any $u\in \Msin$  one has  $\int_\T \sin(u+c)\dx=0$ for $c = 0,\pi \pmod{2\pi}$. If in addition $\int_\T \cos( u)\dx=0$ then $\int_\T \sin( u+c)\dx=0$ for any  $c\in\T_{2\pi}$. We split $\Msin$ into three parts defined as follows
\[
 \Sing=\left\{\;u\in \Leave\colon \int_\T \sin(u)\dx=0, \; \int_\T \cos(u)\dx=0\;\right\}
\]
and
\[
 \Ppm=\left\{\; u\in \Leave \colon \int_\T \sin(u)\dx=0,\;  \pm\int_\T \cos(u)\dx > 0\;\right\}.
\]
\noindent Note that for $k=0$, $u\equiv 0$ is in \Pp and $u\equiv \pi$ is in \Pm and no other constant function is in \Msin for any $k\in\Z$. Furthermore for any $u\in \Leave$, either $\int_\T \sin(u)\dx=0=\int_\T \cos(u)\dx$ in which case $u+c\in\Sing$ for all $c\in\T_{2\pi}$, or at least one of the integrals doesn't vanish in which case there is exactly one $c_+\in \T_{2\pi}$ such that $u+c_+\in\Pp$ and exactly one $c_-\in \T_{2\pi}$ such that $u+c_-\in\Pm$. Note that $c_+$ and $c_-$ are the two solutions $\pmod{2\pi}$ of $\tan(c)=\frac{\int_\T\sin(u)\dx}{\int_\T\cos(u)\dx}$.
Furthermore for $\int_\T \exp(\ii u)\dx \not=0$ one has $\tan(- \arg\left(\int_\T \exp(\ii u)\right)) = \frac{\int_\T\sin(u)\dx}{\int_\T\cos(u)\dx}\pmod{2\pi}$ and $u - \arg\left(\int_\T \exp(\ii u)\right)\in\Pp$, also $\int_\T \exp(\ii u)\dx =0$ iff $\int_\T \sin(u)\dx=0=\int_\T \cos(u)\dx$. This is discussed in more detail in Proposition \ref{prop:separation}.

Now let $l\geq 0$ then for any $u\in \Msin[l+1]$, $u_x/2\in \{\; \mathring{v}+k\pi \in H^{l}\;:\; \mathring{v}\in H^{l}_0\;\}$. Hence in order to parametrize \Pp[l+1] and \Pm[l+1] we extend $\dxinv$ to functions of the form $v = \mathring{v}+k\pi$, where $\mathring{v}\in H^{l}_0$,  by $\dxinv v= \dxinv \mathring{v}+k\pi x\in \Leave$ and introduce, for $l\geq 0$, the spaces
\[
    \Wp=\left\{\;v\in H^{l}: \int_\T \exp(2\ii\dxinv v)\dx\not=0 \text{ and } \int_\T v\dx = k\pi \;\right\}.
\]                                                                                   
The space \Wp is open in the affine space $H_0^l + k\pi$.
Furthermore, for any $v\in H^{l}$ define $\breve{v} \in H^l$ by $\breve{v}(x):=v(-x)$ and note that $\dxinv \breve{v} =-\left(\dxinv v\right)\spbreve$.
The Poisson algebra of equation \eqref{eq:mkdvfocusing} is defined in an analoguous way as the one for equation \eqref{eq:mkdvdefocusing}.

\begin{thm}\label{thm:sin}
For any $l\geq 0$, the following holds:
 \begin{thmenum}  
\item \label{prop:cForSin}
  The maps 
  \[
   \Fsinp:\Wp \to \Pp[l+1], \;v\mapsto 2\dxinv v+\Csin(v), 
   \quad \Csin(v):= -\arg\left(\int_\T \exp(2\ii \dxinv v)\dx\right)
  \]
and
\[
   \Fsinm:\Wp \to \Pm[l+1], \; v\mapsto \left(2\dxinv v\right)\spbreve +\Csin(v) + \pi
  \]
are real analytic parametrizations of \Pp[l+1] and \Pm[l+1] respectively. Furthermore they both transform equation \eqref{eq:sinGordon} into 
 \begin{equation} \label{eq:dxSin}
  v_t=\frac12\sin(2\partial_x^{-1}v+\Csin(v))
\end{equation}
on the space \Wp with vectorfield $\frac12 \sin(2\partial_x^{-1}v+\Csinh(v))\in H_0^{l+1}$.
  \item Equation \eqref{eq:dxSin} is a Hamiltonian PDE with respect to the Gardner bracket with Hamiltonian
  \begin{equation}\label{H:sin} 
	   \Hsin : \Wp\to \R, v\mapsto \frac14\int_\T\cos (2\dxinv v+\Csin(v))\dx.  
  \end{equation}
\item The Hamiltonian $\Hsin$, defined by \eqref{H:sin}, is in the Poisson algebra of the focusing \mkdv equation and hence \eqref{eq:dxSin} is an integrable PDE on \Wp.
 \end{thmenum}
\end{thm}

Theorem \ref{thm:sin} combined with results from \cite{kappeler2008mkdv} yields in particular the following.
\begin{cor}\label{cor:BirkhoffSin} 
For any $l\geq 0$, equation \eqref{eq:dxSin} admits Birkhoff coordinates around 0. In particular, near 0 the Hamiltonian $\Hsin$ is a  function of the actions alone. 
\end{cor}

Furthermore, it turns out  that \Msin is a global fold. Taking into account that $\int_\T\cos(u+\pi)\dx = - \int_\T\cos(u)\dx$ one can prove the following.
\begin{prop}\label{prop:separation}
For any $m\geq 1$, \Sing is a real analytic codimension 2 submanifold of $\Leave$ and hence a codimension 1 submanifold of \Msin, furthermore the map 
\[\Tr:\Msin \to \Msin , u \mapsto \breve{u} +\pi
\]
is a bianalytic involution, mapping \Pp onto \Pm and leaving \Sing invariant which conjugates eqaution \eqref{eq:sinGordon}. 
\end{prop}

\begin{rem*}
Note that $u\equiv 0$ is in \Pp. It follows from Proposition \ref{prop:separation} that there is a neighborhood $U_0$ of $0$ in \Msin such that $U_0\subset\Pp$.
Hence by Corollary \ref{cor:BirkhoffSin} solutions of equation \eqref{eq:sinGordon} for initial value small enough are in  $C^1(\R,\Pp)$.
\end{rem*}

Proposition \ref{prop:separation} can be used to show that equation \eqref{eq:sinGordon} is {\it not} wellposed on \Sing for any $m\geq 1$ in the following sense: 
\begin{thm} \label{thm:nonSolvabilityOnSing}
There is an open  set $U$ in $\Sing[1]$ such that for any $u_0\in U$ and any $T>0$, the initial value problem
\[
\left\{ \begin{array}{l}
	u_{xt} = \sin(u) \\
	u(0)=u_0
\end{array}\right.
\]
admits no solution $u$ in $ C^1([0,T), \Msin[1])$. As a consequence $U$ has nonempty intersection with \Leave  for all $m\geq 1$. This implies that there are elements in $U$ which are $C^\infty$.
\end{thm}

We end this section with a proposition describing \Sing in more details,
\begin{prop} \label{prop:Sing}
Let $u\in\Msin$ and $c\in\T_{2\pi}$ such that $c\not= 0,\pi\pmod{2\pi}$ then:
\begin{thmenum}
\item If $u,u+c\in \Msin$  then $u\in\Sing$.
\item If $u\in\Sing$ then $u+c\in\Sing$.
\item If $u\in \Sing$ then viewing $u$ as a function $u:\R\to\R$ one has
\[
\osc(u):= \sup_{0\leq x,y\leq 1} | u(x)- u(y)|\geq \pi.
\]
\end{thmenum}
\end{prop}

\emph{Application in geometry:}
Equation (\ref{eq:sinGordon}) was already studied in the nineteenth century in the course of investigations of surfaces of constant Gaussian curvature $-1$, also referred to as pseudo spherical surfaces, see e.g. \cite{Chern1986pseudoSpheres}. 
Let $\Sigma\subset \R^2\ni(x,t)$ be an open set. There is a correspondence between any given $C^2$ immersion $r:\Sigma\to\R^3$ of constant Gaussian curvature $-1$ with
$|\partial_xr|=|\partial_tr|=1$, i.e., in Chebyshev form, and solutions $u:\Sigma\to\R$ to \eqref{eq:sinGordon}.
Indeed the first and second fundamental forms of a constant Gaussian curvature -1 immersion in Chebyshev form are given by
\begin{equation} \label{eq:fundamentalForms}
 g=\mathrm{d} x^2+2\cos (u) \mathrm{d} x\mathrm{d} t+ \mathrm{d} t^2\quad \text{and}\quad 
 A=2\sin(u) \mathrm{d} x\mathrm{d} t
\end{equation}
where $u(x,t)$ is the angle between $\partial_xr$ and $\partial_tr$ in $\R^3$ and takes values in $(0,\pi)$. 
Since $g,A$ are the first and second fundamental forms of the immersion, the Gauss-Codazzi-Mainardi equations hold. In terms of $u$ this is
\begin{equation}\label{eq:constCurvature}
\sin(u) = u_{xt}.
\end{equation}
Conversely given any smooth function $u:\Sigma \to (0,\pi)$ define $g,A$ by \eqref{eq:fundamentalForms}. Then $g,A$ fulfill the Gauss-Codazzi-Mainardi equations iff \eqref{eq:constCurvature} holds, in which case by Bonnet's Theorem there is a constant Gaussian curvature $-1$ immersion $\Sigma\to \R^3$ with fundamental forms $g,A$, see e.g.  \cite{galvez2009surfaces}.
Hence all Chebyshev pseudospherical immersions of $\R^2$ which are 1-periodic in $x$ are periodic solutions of equation (\ref{eq:sinGordon}) with values in $(0,\pi)$. According to Corollary \ref{prop:separation}, such a solution is not contained in \Sing.

\smallskip
\emph{Method of proof and related results:}
Various techniques for obtaining solutions of the complex sine-Gordon equation on the real line have been developed, amongst them the B\"acklund transformation, inverse scattering and the transformation $u=4 \atanh (w)$, see e.g.  \cite{SolitonBook}.
Solutions to equation (\ref{eq:sinGordon}) have been constructed in \cite{Ablowitz1973SineLigthLine} and \cite{Belokolos1994AlgebroGeometric}, using the algebro-geometric approach. Under appropriate conditions solutions can be represented in terms of theta functions.

In \cite{Chodos1980SineKdv}, Chodos  related the standard Lax pair of the \mkdv equation and the one of sine-Gordon equation and showed in this way that the two equations have infinitely many conserved quantities in common. More precisely, he showed that for the well known Lax pair formulation of \kdv
\[
	\partial_t L_\kdv = [L_\kdv,B_\kdv],
\]
and the Lax pair formulation of equation \eqref{eq:complexSine}
\[
 \partial_t L_{SG} = [ L_{SG},B_{SG}]
\]
there is a map mapping the periodic and anti-periodic spectra of $L_{SG}$ onto the corresponding ones of $L_\kdv$ at $q= u_{xx}+u_x^2$.  By the Lax pair formulation of equation \eqref{eq:complexSine}, the periodic and anti-periodic spectra of $L_{SG}$ are constants  of motion and hence so are the ones of $L_\kdv$. In \cite{kappeler2008mkdv} a new Lax pair formulation for the \mkdv equation,
\[
	\partial_t L_M = [L_M,B_M].
\]
is introduced and a similar map constructed mapping the spectra of $L_M$ at $v$ onto the one of $L_\kdv$ at $v_x+v^2$. Hence there is a mapping from the spectra of $L_{SG}$ at $u$ onto the spectra of $L_M$ at $u_x$.
One of the main contributions of our paper is to write \Msin as a branched double cover over $H^m_0$, with branches $\Pp,\Pm$ and ramification locus \Sing, and then to parametrize \Ppm conveniently leading to equation \eqref{eq:dxSin} and to show that the latter is a Hamiltonian PDE when viewed on the phase space of the focusing \mkdv equation. Surprisingly, to the best of our knowledge, the fact that \eqref{eq:sinGordon} is not well posed on the ramification locus \Sing was not known before.

\emph{Organisation:} 
In Section 2 we prove items $(i)$ and $(ii)$ of Theorems \ref{thm:Sinh} and \ref{thm:sin} as well as Proposition \ref{prop:separation} and \ref{prop:Sing}. Item $(iii)$ of Theorems \ref{thm:Sinh} and \ref{thm:sin} as well as Corollaries \ref{cor:BirkhoffSinh} and \ref{cor:BirkhoffSin} are proved in Section 3. In Section 4 we prove Theorem \ref{thm:nonSolvabilityOnSing}.

\emph{Acknowledgement:} The author is very grateful to Professor Thomas Kappeler
for sharing expertise, and sincere and valuable guidance and encouragement.
\section{Proofs: part 1}

\begin{proof}[Proof of Theorem \ref{thm:Sinh} $(i)$] First note that
\[	
	\int_\T\cosh(2\dxinv v)\dx \geq 1 \quad\ \text{and} \quad 
	\left|\int_\T\sinh(2\dxinv v)\dx \right| < \int_\T\cosh(2\dxinv v)\dx,
\]
imply that \Csinh is well defined. By the addition theorem for sinh, one sees that 
\begin{equation}\label{eq:cSinh}
 \int_\T \sinh(2\dxinv v+\Csinh(v))\dx =0.
\end{equation}
Hence the map \Fsinh takes indeed values in \Msinh[l+1]. Clearly $\Csinh:H_0^0 \to \R$ is real analytic and so is its restriction to $H^l_0$. One then easily verifies that $\Fsinh:H_0^l\to \Msinh[l+1]$ is a real analytic diffeomorphism onto its image. 
To see that \Fsinh is onto, note that any $u\in \Msinh[l+1]$ can be written as $2\dxinv v+c$ where $c=\int_\T u(x)\dx$ and $v\in H_0^l$. Since $\sinh(x)$ is strictly increasing on $\R$, 
\[
	\int_\T\sinh(2\dxinv v+c)\dx = 0
\]
has the unique solution $\Csinh(v)$ and hence $\Fsinh(v)=u$, proving that \Fsinh is onto. 
In a straight forward way one verifies that \Fsinh transforms \eqref{eq:sinhGordon} into \eqref{eq:dxSinh}.
\end{proof}

\begin{proof}[Proof of Theorem \ref{thm:Sinh} $(ii)$]
We compute the Hamiltonian vector field $\partial_x\partial \Hsinh$ of \Hsinh with respect to the Gardner bracket
\begin{align*}
\partial_x\partial \Hsinh &=\partial_x\partial\left( -\frac14\int_\T\cosh(2\dxinv v + \Csinh(v))\dx\right)
\\&=-\frac14\partial_x \left(-2\dxinv\sinh(2\dxinv v+\Csinh(v)) +\partial \Csinh(v)\int_\T\sinh(2\dxinv v+\Csinh(v))\dx   \right)
\\&=\frac12\sinh(2\dxinv v+\Csinh(v)).
\end{align*}
\end{proof}

\begin{proof}[Proof of Proposition \ref{prop:separation}]
For any $m\geq 0$, $T$ is obviously a bianalytic diffeomorphism. Furthermore for $u\in \Leave$, we have
\[
 \int_\T \sin(\breve u+ \pi)\dx = -\int_\T \sin(u)\dx \quad \text{and}\quad  \int_\T \cos(\breve u+ \pi)\dx = -\int_\T \cos(u)\dx
\]
implying that $T$ leaves \Sing invariant and maps \Pp onto \Pm. For the second part note that 
\[
\left. T\right|_{\Pp} = \Fsinm\circ {\Fsinp}^{-1}
\] 
since $\Fsinp$ and $\Fsinm$ both transform equation \eqref{eq:sinGordon} to equation \eqref{eq:dxSin}, hence $T$ conjugates equation \eqref{eq:sinGordon}.
\end{proof}

\begin{proof}[Proof of Proposition \ref{prop:Sing} $(i)$]
Define 
\[
   K:\Leave \to \R^2, u\mapsto   \int_\T e^{\ii u}\dx
\]
then $\Ppm = K^{-1}(\pm \R_{>0})$ and $\Sing = K^{-1}(0)$, furthermore $K(u+c)= e^{\ii c}K(u)$. Hence $u,u+c \in \Msin$ iff $K(u)= 0$ or $e^{\ii c}=\pm 1$.
\end{proof}

\begin{proof}[Proof of Proposition \ref{prop:Sing} $(ii)$]
For $u\in\Sing$ one has $K(u)=0$ this implies $K(u+c)=e^{\ii c}K(u) = 0$ for all $c\in\T_{2\pi}$.
\end{proof}

\begin{proof}[Proof of Proposition \ref{prop:Sing} $(iii)$]
Assume that $u\in \Sing$ has oscillation smaller than  $\pi$. We can assume by $(iv)$ that $\min_{x\in[0,1]}(u(x))=0$ hence $\left.u\right|_{[0,1]}$ takes values in $[0,\pi)$ and $\sin(u)$ is non negative.  As $\int_\T\sin(u)\dx=0$ it implies that $u\equiv 0$ and hence $\int_\T\cos(u)\dx=1$ which is a contradiction.
\end{proof}

\begin{proof}[Proof of Proposition \ref{prop:separation}: \Sing is a real analytic submanifold of \Leave of codimension 2]
To prove that \Sing is a real analytic submanifold note that the differential of $K$ is
\[
	 d_uK(h) = \int_\T \ii\cos(u)h - \sin(u)h \dx.
\]
It follows that $dK$ has vanishing components only for $u\equiv 0\pmod{2\pi}$ or $u\equiv \pi \pmod{2\pi}$. Hence $dK$ does not vanish on \Sing and by $(v)$ we can find $h_1,h_2\in H^{m}$ such that
\[
\left|\int_\T\cos(u)h_1\dx\right|>\left|\int_\T\sin(u)h_1\dx\right| \quad \text{and}\quad
\left|\int_\T\cos(u)h_2\dx\right|<\left|\int_\T\sin(u)h_2\dx\right|
\]
hence $d_uK(h_1)$ and $d_uK(h_2)$ are linearly independent over $\R$ hence by the implicit function theorem \Sing is a realanalytic codimension two submanifold of \Leave. 
\end{proof}

\begin{proof}[Proof of Theorem  \ref{thm:sin} $(i)$]
By definition of \Wp one has $K(2\dxinv v) \not=0$ for all $v\in \Wp$ and $\arg:\C\to\R/2\pi\Z$ is well defined such that
\[
K(2\dxinv+c) = e^{\ii c}K(2\dxinv v)\in \R_{>0} \quad \Leftrightarrow \quad c = -\arg(K(2\dxinv v))).
\]
Hence \Csin is well defined and the map \Fsinp takes values in \Pp[m+1] while \Fsinm takes values in \Pm[m+1]. By the implicit function theorem and as
\[
	\left.\partial_c\right|_{c=\Csin(v)} \int_\T\sin(2\dxinv v+c)\dx  =	\int_\T\cos(2\dxinv v+\Csin(v)) >0
\]
\Csin is real analytic on $W^0$ and so is its restriction to \Wp for any $m\geq0$. One then easily verifies that \Fsinp and \Fsinm  are real analytic diffeomorphisms onto their image. To see that \Fsinp is onto, note that any $u\in \Pp$ can be written as $2\dxinv v+c$ where $c\in\R/2\pi\Z$ and $v\in \Wp$. Since $e^{\ii c}K(2\dxinv v) = K(2\dxinv v+c)= K(u)\in \R_{>0}$ $c=\Csin(v)\pmod{2\pi}$
and hence $\Fsinp(v)=u$, proving that \Fsinp is onto. Similarly one shows that \Fsinm is onto.
Finally one verifies in a straight forward way, that \Fsinp and \Fsinm transform \eqref{eq:sinGordon} on \Pp[m+1] respectively \Pm[m+1] into equation \eqref{eq:dxSin} on  \Wp.
\end{proof}

\begin{proof}[Proof of Theorem  \ref{thm:sin} $(ii)$]
First note that at any point $v\in\Wp$ the tangent space of \Wp is $H_0^m(\T,\R)$. The derivative $\partial_v \Hsin$ in direction $h\in H_0^m(\T,\R)$ can be computed as follows
\begin{align*}
\partial_v \Hsin(h) &=\partial_v\left( \frac14\int_\T\cos(2\dxinv v + \Csin(v))\dx\right)(h)
\\&=\frac14\int_\T-\sin(2\dxinv v+\Csin(v))\left(2\dxinv h + \partial_v \Csin(h)\right)\dx
\end{align*}
using that  $\int_\T\sin(\dxinv v+\Csin(v))\dx=0$ and that $\partial_v \Csin(h)$ is independent of $x$ one has
\begin{align*}
\partial_v \Hsin(h) &=\frac14\int_\T-\sin(2\dxinv v+\Csin(v))2\dxinv h\dx 
\\&=\frac12\int_\T \dxinv\sin(2\dxinv v+\Csin(v))h\dx - \left.\frac12\sin(2\dxinv v+\Csin(v))\dxinv h\right|_{x=0}^{x=1}.
\end{align*}
As $\sin(\dxinv v+\Csin(v))$ is 1-periodic in $x$ and $\left.\dxinv h\right|_{x=0}^{x=1}=0$ one concludes
\begin{align*}
\partial_v \Hsin(h)  &=\frac12\int_\T \dxinv\sin(2\dxinv v+\Csin(v))h\dx.
\end{align*}
As a consequence
\begin{align*}
\partial_x\partial \Hsin &=\partial_x\frac12\dxinv\sin(2\dxinv v+\Csin(v))
=\frac12\sinh(2\dxinv v+\Csin(v)).
\end{align*}
\end{proof}

\section{Proofs: part 2}
To construct Birkhoff  coordinates for the Hamiltonians  \eqref{H:sinh} and \eqref{H:sin}, we use results for the  \mkdv equation and its hierarchy in \cite{kappeler2008mkdv}. In order to state the theorem on Birkhoff coordinates for \mkdv   from \cite{kappeler2008mkdv}, we introduce for any  $\alpha\in\R$ $h_\alpha=\ell_\alpha^2\times \ell_\alpha^2$, where $\ell_\alpha^2=\ell_\alpha^2(\N,\R)$ is the weighted $\ell^2$ sequence space given by
 \[
  \ell_\alpha^2:=\left\{\;x\in \ell^2(\N,\R)\;:\;\norm{x}_\alpha=\left(\sum_{k\geq 1}k^{2\alpha}|x_k|^2\right)^{1/2}<\infty\;\right\},
 \]
it is endowed with the standard Poisson bracket where $\{x_n,y_k\}=-\{y_k,x_n\}=\delta_{nk}$, whereas all other brackets between coordinate functions vanish. 

The defocusing \mkdv is a Hamiltonian PDE on $H^1$ with respect to the Gardner bracket and Hamiltonian 
\[
 H^\text{mkdv} :H^1\to \R, \quad v\mapsto\frac12 \int_\T v_x^2+v^4 \mathrm{d}x.
\]

\begin{thm*}[\cite{kappeler2008mkdv}] \label{thm:mkdvBirkhoff} There exists a map
\[
\Omega : H^1 \to h_{3/2} \times \R,  v \mapsto ((x_k , y_k )_{k\geq 1} , [v])
\]
with the following properties:
\begin{enumerate}
 \item for any $m \geq 1$, $\Omega: H^m \to h_{m +1/2} \times \R$ is a real analytic diffeomorphism;
 \item $\Omega$ is canonical, i.e. it preserves the Poisson brackets;
 \item $H^\mkdv\circ \Omega^{-1}$ only depends on the action variables, $I_k := (x_k^2 + y_k^2 )/2, k \geq 1$, and the value of the Casimir $[v]$, i.e. it is in Birkhoff normal form.
 \end{enumerate}
The map $\Omega$ is referred to as Birkhoff map and $(x_k , y_k )_{k\geq 1}$ as
Birkhoff coordinates for $\mkdv$.
\end{thm*}

\begin{rem*}
To extend this result to the case $m=0$ one can follow the proof in \cite{kappeler2008mkdv} and use results for \kdv in \cite{kappeler2005birkhoffHm1KdV}.
\end{rem*}

These coordinates have the property that all Hamiltonian PDE's in the Poisson algebra of equation \eqref{eq:mkdvfocusing} transform into a Hamiltonian PDE in Birkhoff normal form. Furthermore as the mapping is real analytic the complexification of the Hamiltonian system with Hamiltonian
\[
H_\C^\text{mkdv} :H^m(\T,\C)\to \C, \quad v\mapsto\frac12 \int_\T v_x^2+v^4 \mathrm{d}x
\]
is in normal form in a complex neighborhood of $H^m$ in $H^m(\T,\C)$, since the Lax-Pair formulation is still valid in this case. Note that the focusing \mkdv is the restriction of this system to the invariant subspace $H^m(\T,\ii \R)$, this implies, that the focusing \mkdv equation, and hence its Poisson algebra, has Birkhoffcoordinates around 0.

To prove item $(iii)$ of Theorems \ref{thm:Sinh} and \ref{thm:sin} we write equation \eqref{eq:complexSine} as a Hamiltonian PDE for $V = \partial_x U$ on a neighborhood $\nbhd$ of $H_0^m$ and $\Wp$ in $H_0^m(\T,\C)$ with Hamiltonian 
\[
H_{SG}(V) = \int_{\T} \cosh(2\dxinv V +C(V))
\]
where $C(V)\in \C$ is the solution$\pmod{2\pi\ii}$ of 
	$\int_\T \sinh(2\dxinv V+c)\dx =0$
such that 
\[
	\Re \left(\int_\T \cosh(2\dxinv V+c)\dx\right) \geq 0.
\]
Note that the Lax formulation of the sine-Gordon equation is also valid in the complex case and that the Hamiltonians \ref{H:sinh} and \ref{H:sin} are the restrictions of $H_{SG}$ to the invariant subspaces $H^m_0$ respectively $\ii \Wp\subset H^m(\T,\ii\R)$. Where  $H^m_0$ is the restriction of the defocusing \mkdv phase space to mean zero and $\ii \Wp\subset H^m(\T,\ii\R)$ is open in the restriction, of the phase space of the focusing \mkdv, on to $k\pi$ mean valued functions. Since the mean value is a Casimir of the \mkdv Poisson structure these are symplectic.
To finally show that $H_{SG}$ is in the $\mkdv$ Poisson algebra we need to recall the results in \cite{Chodos1980SineKdv} and \cite{kappeler2008mkdv} more precisely.
For the operators $L_{SG},L_\kdv$ and $L_M$, for arbitrary potentials $U,V$ and $Q$ in the corresponding phase spaces, define $\Dl_{SG}(U,\lm),\Dl_\kdv(V,\lm)$ and respectively $\Dl_M(Q,\lm)$ to be the Discriminant of the corresponding operators, this is the trace of the corresponding Floquet matrices. It is well known that the discriminants determine the periodic and anti-periodic spectra. Now the result of Chodos states that 
\[
	\Dl_\kdv(U_{xx}+ U_x^2,\lm^2) = \Dl_{SG}(U,\lm)
\]
further in \cite{kappeler2008mkdv} it is shown that
\[
	e^{\int_\T V\dx}\Dl_M(V,\lm) = \Dl_\kdv(V_x+V^2,\lm)
\]
hence
\[
	e^{\int_\T U_x\dx}\Dl_M(U_x,\lm^2) =\Dl_{SG}(U,\lm).
\]
By 	the Lax formulation for the sine-Gordon equation the periodic spectrum of $L_{SG}(U)$ is preserved by the flow hence the periodic spectrum of $L_M(U_x)$ is a preserved quantity. This implies that the Hamiltonain \ref{H:sinh} is in the Poisson algebra of the defocusing \mkdv and hence shares the same Birkhoff map, which proves Theorem \ref{thm:Sinh} $(iii)$ and Corollary \ref{cor:BirkhoffSinh} further the Hamiltonian \ref{H:sin} is in the Poisson algebra of the focusing \mkdv and the Birkhoff map near 0 for the focusing \mkdv is a Birkhoff map for equation \ref{eq:dxSin} which is the proof of Theorem \ref{thm:sin} $(iii)$ and Corollary \ref{cor:BirkhoffSin}.\qed

\section{Proofs: part 3}
In this section we prove Theorem \ref{thm:nonSolvabilityOnSing}.
Let us first assume that for some $T>0$ $u\in C^1([0,T),\Leave[1])$ is a solution of \eqref{eq:sinGordon}.
 Note that $u_t = \dxinv \sin(u) + \partial_t[u]$, where $[u] = \int_\T u\dx$, implies that
\[
 0= \partial_t \int_\T \sin(u)\dx=-\int_\T \cos(u)u_t\dx
 = -\int_\T \cos(u)\dxinv \sin(u)\dx+\int_\T \cos(u)\partial_t[u]\dx
\]
defining $K_1(u) =\int_\T \cos(u)\dxinv \sin(u)\dx$ and using that $\partial_t[u]$ is  independent of $x$ we get the constraint
\begin{equation}\label{eq:singdt}
	-K_1(u) + \partial_t[u]\int_\T \cos(u)\dx=0.
\end{equation}
For a $C^1-$solution $v$ of equation \eqref{eq:dxSin} one can compute the derivative of $\Csin$ by the implicit function theorem and then substituting \eqref{eq:dxSin} into \eqref{eq:singdt}
one obtains
\[
	-K_1(2\dxinv v+\Csin(v))+ \partial_v\Csin(\partial_tv)\int_\T \cos(2\dxinv v+\Csin(v))\dx=0
\]
hence \eqref{eq:singdt} does not give further information about the solution $v$. On the other hand when $u\in\Sing$ one has $\int_\T \cos(u)\dx = 0$ which leads to:
\begin{lem}
For any $u_0$ in \Sing such that there exists $T>0$ and a solution $u\in C^1([0,T),\Leave[1])$ with $u(0) = u_0$, one has $K_1(u_0)=0$.
\end{lem}
\begin{proof}[Proof of Theorem \ref{thm:nonSolvabilityOnSing}]
To prove Theorem \ref{thm:nonSolvabilityOnSing} we are left to show the following lemma.
\begin{lem}\label{lem:k3}
There is an open set $U$ in \Sing containing elements of any topological charge $k\in\Z$, such that $K_1(u) \not= 0$ for all $u\in U.$
\end{lem}
In the rest of this section we will prove Lemma \ref{lem:k3}. Since $K_1$ is continuous it is enough to show that for any $k\in\Z$ there is an element in \Sing which has topological charge $k$ such that $K_1(u)\not = 0$. Define for $k\in\Z$ and $r\in(0,1)$
\[
  w_k(x) = 
  \left\{ \begin{array}{lr}
		2k\pi \frac{x}{r} & x<r\\
		  2k\pi -2\pi \frac{x-r}{1-r} & x\geq r
	\end{array}\right.
\] 
Note that for $w_k$ $K(u)=0$ hence $w_k$ is in \Sing[1] further the topological charge of $w_k$ is $(k-1)$ and the following holds 
\begin{equation*}
\cos(w_0) = \left\{ \begin{array}{lr}
		\cos(2k\pi \frac{x}{r}) & x<r\\
		  \cos(2\pi\frac{x-r}{1-r}) & x\geq r
	\end{array}\right.
\quad
\dxinv \sin(w_0) =\left\{ \begin{array}{lr}
		-\frac{r}{2k\pi}\cos(2k\pi \frac{x}{r}) +(1-r)\frac{k(1-r)+r}{2k\pi}& x<r\\
		  \frac{1-r}{2\pi}\cos(2\pi \frac{x-r}{1-r}) -r\frac{k(1-r)+r}{2k\pi} & x\geq r
	\end{array}\right.
\end{equation*}
And one computes using that $\int_0^{r}\cos(w_0)\dx=\int_{r}^1\cos(w_0)\dx=0$,
\[
	K_1(w_0) = -\frac{r}{2k\pi}\int_0^{r}\cos(2k\pi \frac{x}{r})^2\dx+\frac{1-r}{2\pi} \int_{r}^1\cos(2\pi \frac{x-r}{1-r})^2\dx
\]
which by the substitutions $s=x\frac{k}{r}$  respectively $s=\frac{x-r}{1-r}$ gives
\[
 K_1(w_0) = -\frac{r^2}{2k^2\pi}\int_0^k\cos(2\pi s)^2\ds + \frac{(1-r)^2}{2\pi}\int_0^1\cos(2\pi s)^2\ds = \left(\frac{1-r}{2\pi}-\frac{r^2}{2k\pi}\right)\int_0^1\cos(2\pi s)\ds
\]
hence since $\int_0^1\cos(2\pi s)\ds \not = 0$ one has $K_1(w_0)=0$ iff
\[
	k(1-r)-r^2=0
\]
which for a given $k$ is only solved by one $r\in (0,1)$.
\end{proof}

\begin{rem}
One can read off from the proof of Theorem \ref{thm:nonSolvabilityOnSing} that for any sequences $0=t_0<t_1<\dots <t_n=1$ and $(k_i)_{1\leq i\leq n}\subset \Z$, with $n\in\N$ arbitrary such that, the function $u\in \Leave[1]$, defined by 
\[
	\left. u\right|_{[t_i,t_{i+1}]}(t) = 2\pi k_i\frac{t_{i+1}-t}{t_{i+1}-t_i} + 2\pi k_{i+1}\frac{t-t_i}{t_{i+1}-t_i},
\]
is in \Sing[1]. Furthermore, one can show that for any $(k_i)_{1\leq i\leq n}\subset\Z$, the set of sequences $(t_i)_{1\leq i\leq n}$ as above with $K_1(u)=0$ has Hausdorff dimension $n-1$. Most likely, the subset where $K_1(u)=0$, is a submanifold of \Sing. On this submanifold one could repeat the arguments of Theorem \ref{thm:nonSolvabilityOnSing} and take the time derivative of $K_1(u(t))$ to get another constraint of the form $K_2(u)=0$. Iterating this procedure one might obtain infinitely many additional constraints.
\end{rem}

\bibliographystyle{acm}
\bibliography{literatur.bib}

\end{document}